\documentclass[preprint,twocolumn,5p]{elsarticle}

\usepackage{amssymb}

\journal{Systems \& Control Letters}

\usepackage{graphicx}          

\usepackage[utf8]{inputenc}

\usepackage{epsfig} 
\usepackage{amsmath} 
\usepackage{bm}
\usepackage{graphicx}
\usepackage{amsmath}
\usepackage{xcolor}
\usepackage[english]{babel}
\usepackage{amsfonts}
\usepackage{amssymb}
\usepackage{amsbsy}
\usepackage{hyperref}
\usepackage{epstopdf}
\usepackage{caption}
\usepackage{subcaption}
\usepackage{tabularx}
\usepackage{dsfont}
\usepackage{amsopn}
\usepackage{mathtools}
\usepackage{enumerate}  
\usepackage{balance} 

\usepackage{irt-common}
\usepackage{irt-shortcuts}

\newtheorem{theorem}{Theorem}
\newtheorem{lemma}{Lemma}
\newtheorem{corollary}{Corollary}

\newdefinition{definition}{Definition}
\newdefinition{remark}{Remark}
\newproof{proof}{Proof}

\begin{document}
	
	\begin{frontmatter}

		\title{Strong Detectability and Observers for Linear Time Varying Systems} 
		
		 \author[irt]{Markus Tranninger\corref{cor1}}\ead{markus.tranninger@tugraz.at}
		 \author[cdlab]{Richard Seeber}\ead{richard.seeber@tugraz.at}
		 \author[btu]{Juan G. Rueda-Escobedo}\ead{ruedaesc@b-tu.de}
		 \author[irt,cdlab]{Martin Horn}\ead{martin.horn@tugraz.at}
		 
		 \address[irt]{Institute of Automation and Control, Graz University of Technology, Inffeldgasse 21B/II, 8010 Graz, Austria}

		 \address[cdlab]{Christian Doppler Laboratory for Model Based Control of Complex Test Bed Systems, Institute of Automation and Control, Graz University of Technology, Inffeldgasse 21B/II, 8010 Graz, Austria}
	             
	    \address[btu]{Fachgebiet Regelungssysteme und Netzleittechnik, Brandenburgische Technische Universität Cottbus - Senftenberg, Siemens-Halske-Ring 14, 03046 Cottbus, Germany}
    
    	\cortext[cor1]{corresponding author}

	
\begin{abstract}                          
This work presents a notion of strong detectability for linear time varying systems affected by unknown inputs.
It is shown that this notion is equivalent to detectability of an auxiliary system without unknown inputs.
This allows a straightforward observer design for dependable state estimation in the presence of unknown inputs.
The design reduces to a deterministic Kalman filter design problem, where the observer gains can be obtained from the solution of a differential Riccati equation.
The efficacy of the proposed approach is demonstrated by means of a numerical example.
\end{abstract}

\end{frontmatter}

\section{Introduction}\label{sec:introduction}

State observers are important building bricks for feedback loop design and fault detection.
By means of inputs and measurable outputs of the system, the non-measurable system states are estimated.
In practical applications, observers have to deal with the fact that the measured signals are affected by unknown external disturbances. 
One strategy to cope with such disturbances is to model them as unknown system inputs and utilize unknown input observers~\cite{hautus1983strong,valcher1999state}.
Important applications for unknown input observers are, e.g., robust fault detection filter design~\cite{chen1998robust} or, more recently, cyber attack detection in distributed networked control systems~\cite{gallo2020adistributed}. 
The problem of state estimation in the presence of unknown inputs is closely related to the concept of strong detectability.
For linear time invariant systems, strong detectability is a well studied system property, which is extensively treated in~\cite{hautus1983strong,schumacher1980ontheminimal}. 
The developed theory, which is required for the design of unknown input observers~\cite{valcher1999state,hou1994fault} has found a wide variety of applications, see, e.g.~\cite{chen1998robust,gallo2020adistributed,edwards1998sliding}.

Hautus introduced the notions of strong observability, strong detectability and strong$^*$ detectability in a linear time invariant setting~\cite{hautus1983strong}.
The focus of the present work is on the latter notion, because in the linear time invariant case, strong$^*$ detectability is necessary and sufficient for the existence of an unknown input observer.
Moreover, neither strong detectability nor strong observability imply strong$^*$ detectability, see~\cite{hautus1983strong}.

One of the few contributions related to the aforementioned notions in the linear time varying case is~\cite{kratz1998alocal}.
There, conditions for strong observability, which are based on successive derivatives of the output signal, are presented.
These pointwise-in-time conditions, however, do not guarantee the existence of a linear observer and require the availability of the output's derivatives.
In~\cite{galvan-guerra2016highorder,tranninger2018sliding}, these derivatives are obtained by utilizing a robust exact higher order sliding mode differentiator.
The applicability of that approach, however, may be limited in the presence of measurement noise, due to the appearing higher order derivatives of the output signal.
Moreover, the unknown input is required to be differentiable and uniformly bounded.
In the present work, these properties of the unknown input as well as derivatives of the output signal are not required.

The present work extends the notion of strong$^*$ detectability to the linear time varying setting. 
This extension also guarantees the existence of a linear observer.
The necessary and sufficient existence conditions for this observer are stated in terms of classical detectability conditions for an auxiliary linear time varying systems.
It is moreover shown that the observer design in fact can be reduced to a classical observer design problem without unknown inputs.
The required observer gain can then be obtained from the solution of a differential Riccati equation. 
Moreover, different properties of strongly detectable systems are investigated based on their time invariant counterpart.

The paper is structured as follows:
The considered problem is stated Section~\ref{sec:problemstatement}.
Preliminaries related to classical detectability concepts are introduced in Sec.~\ref{sec:detectability}.
The main results, i.e., the concept of strong$^*$ detectability for linear time varying systems together with the proposed observer design, are introduced in Section~\ref{sec:mainresult}.
The obtained results are utilized for an unknown input observer design for a linearized version of the Lorenz'96 model in Section~\ref{sec:numericalexample} and the efficacy of the proposed approach is shown by numerical results.
Section~\ref{sec:discussion} draws conclusions and points out future research directions.

\emph{Notation:} Matrices are printed in boldface capital letters, whereas column vectors are boldface lower case letters. 
$\M^\dag$ denotes the Moore-Penrose pseudoinverse of the matrix $\M$.
The matrix $\I_n$ is the $n\times n$ identity matrix.
The 2-norm of a vector or the corresponding induced matrix norm is denoted by $\|\cdot\|$.
Symmetric positive definite (positive semidefinite) matrices $\M^\transp = \M$ are denoted by $\M \succ 0$ ($\M \succeq 0$). 
If for two symmetric matrices $\M_2 - \M_1 \succ 0$ ($\succeq 0$), then $\M_1\prec \M_2$ ($\M_1\preceq \M_2$).

In dynamical systems, differentiation of a vector $\x$ with respect to time $t$ is expressed as $\dot \x$.
Time dependency of state (usually $\x$) and input (usually $\w$) is omitted, whereas the time dependency of the system parameters is stated explicitly.

\section{Problem Statement}\label{sec:problemstatement}
Consider the linear time varying system
\begin{subequations}\label{eq:system}
\begin{align} 
	\dot \x &= \A(t)\x+\D(t)\w,\label{eq:systemstate} \\
	\y&= \C(t)\x \label{eq:systemoutput}
\end{align}
\end{subequations}
for $t\in\mathds J=[0,\infty)$ with the state $\x(t)\in\mathds{R}^n$, the output $\y(t)\in\mathds{R}^p$ and the unknown input $\w(t)\in\mathds{R}^q$. The matrices $\A(t)$, $\D(t)$ and $\C(t)$ of appropriate dimensions are assumed to be uniformly bounded. The goal is to asymptotically reconstruct the system state $\x(t)$ without knowledge of the input $\w(t)$ and to state conditions for that purpose.
It should be noted that no assumptions about differentiability or boundedness of the unknown input $\w(t)$ are made.

Without loss of generality, no known inputs are considered in system~\eqref{eq:system}, because they can always be eliminated and hence do not influence the observer error dynamics.

The following assumptions are made throughout the paper:
\begin{enumerate}
	\item[(A1)] The matrices $\A(t)$ and $\D(t)$ are uniformly bounded and differentiable with uniformly bounded derivatives
	\item[(A2)] The matrix $\C(t)$ is uniformly bounded and at least two times differentiable with uniformly bounded derivatives
	\item[(A3)] There exists a scalar $\gamma >0 $ such that $\mtGamma(t) =\C(t)\D(t)$ fulfills
	\begin{equation}
	\mtGamma^\transp(t) \mtGamma(t) \succeq \gamma \I_q.
	\end{equation}
\end{enumerate}
\textbf{Remark:} Assumption (A3) requires $p \ge q$ and is a uniform extension of the well known rank condition $\rank \C\D =\rank \D = q$, which is a necessary condition for strong$^*$ detectability in the linear time-invariant setting~\cite[Theorem 1.6]{hautus1983strong}.

\section{Preliminaries -- Detectability}\label{sec:detectability}
For stability analysis, the notions of asymptotic stability and uniform exponential stability, see, e.g. \cite{anderson2013stabilizability}, are relevant.
\begin{definition}[stability]
   \label{def:stability}
The linear system \mbox{$\dot \x = \A(t) \x$} is called
\begin{enumerate}[i)]
	\item
	\emph{asymptotically stable (AS)}, if its origin is Lyapunov stable and \mbox{$\lim_{t \to \infty} \x(t) = \bm{0}$} is satisfied for every solution $\x(t)$;
	\item
	\emph{uniformly exponentially stable (UES)}, if there exist positive constants $\mu$ and $M$ such that
	\begin{equation}\label{eq:ues_bound}
		\|{\x(t)}\| \le M e^{-\mu (t - t_0)} \|{\x(t_0)}\|
	\end{equation}
	holds for every solution $\x(t)$ and every $t_0$.
\end{enumerate}
\end{definition}

If the input is identically zero, i.e., $\w(t) = \vc 0$ for all $t\in\mathds J$, the condition for a successful observer design reduces to (uniform) detectability, which is introduced next.

\begin{definition}[uniform exponential detectability]
	System~\eqref{eq:system}, or equivalently the pair $(\A(t),\C(t))$, is called uniformly exponentially detectable, if there exists a uniformly bounded $n\times p$ matrix $\L(t)$ such that the system
	\begin{equation}\label{eq:error:dynamics}
	\dot\e = \left[\A(t)-\L(t)\C(t)\right]\e
	\end{equation}
	is uniformly exponentially stable.
\end{definition}
In the following, detectability refers to the above definition.

The pair $(\A(t),\C(t))$ is called (uniformly exponentially) stabilizable, if its corresponding dual system is detectable.
Details regarding duality between stabilizability and detectability can be found, e.g., in \cite{mueller2010normalized}.

Uniform exponential detectability is equivalent to the existence of an observer of the form 
\begin{equation}\label{eq:observer}
    \dot{\hat\x} = \A(t)\hat\x + \L(t)\left[\y-\C(t)\hat \x\right],
\end{equation}
where $\hat\x(t)$ is the state estimate, which induces the UES error dynamics~\eqref{eq:error:dynamics} with $\e(t)=\x(t)-\hat\x(t)$.
Often, this observer is designed with a (deterministic) Kalman filtering approach~\cite[Sec. 8.3]{sontag2013mathematical}.
In fact, detectability is equivalent to the existence of a uniformly bounded positive semidefinite solution of the filtering Riccati equation as summarized in the following, see~\cite[Lemma 3.4]{ravi1992normalized} and its corresponding proof.
\begin{lemma}[detectability and Riccati equation]\label{le:detectability}
	Let $\Q(t)\succeq 0$ such that $(\A(t),\Q^{\frac{1}{2}}(t))$ is stabilizable.
	Then, the pair $(\A(t),\C(t))$ is detectable if and only if there exists a uniformly bounded positive semidefinite solution $\P(t)$ to the Riccati equation
	\begin{equation}
		\dot \P = \A(t)\P + \P\A^\transp(t) - \P\C^\transp(t)\C(t)\P + \Q(t),
	\end{equation}
	Moreover,~\eqref{eq:error:dynamics} is uniformly exponentially stable with 
	\begin{equation}
		\L(t)=\P(t)\C^\transp(t).
	\end{equation}
\end{lemma}
It should be remarked that the matrix $\Q(t)$ can be regarded as a design parameter and hence stabilizability of $(\A(t),\Q^{\frac{1}{2}}(t))$ can always be achieved.

In the time invariant case, detectability means that a vanishing output implies a vanishing state.
The following lemma shows that the considered detectability notion yields a similar property.
\begin{lemma}[output convergence]\label{le:det:convergence}
	If system~\eqref{eq:system} is detectable, then, $\lim_{t\rightarrow\infty} \y(t)= \bf 0$ for $\w(t)=\bm 0$ implies $\lim_{t\rightarrow\infty} \x(t) =\bf 0$.
\end{lemma}
\begin{proof}
Detectability implies the existence of a uniformly bounded $\L(t)$ such that~\eqref{eq:observer} is an observer for~\eqref{eq:system} and~\eqref{eq:error:dynamics} is UES.
Hence, one obtains
\begin{equation}\label{eq:observer:outputinjection}
	\dot{\hat\x}= \left[\A(t)-\L(t)\C(t)\right] \hat \x + \L(t)\C(t)\y.
\end{equation}
According to~\cite[Theorem 59.1]{hahn1967stability}, it then follows from $\y(t)\rightarrow \bf 0$ that $\hat \x(t) \rightarrow \bf 0$ and hence $\x(t)\rightarrow \bf 0$.\hfill\qed
\end{proof}
In the time invariant case, also the converse of Lemma~\ref{le:det:convergence} holds. 
This is not true in the time varying case.
To see this, consider the system
\begin{equation}
	\dot x =-\frac{1}{t+ 1}x,\quad y=0,
\end{equation} 
which is asymptotically stable, and hence the above implication is trivially fulfilled.
For this system, however, no observer of the form~\eqref{eq:observer} exists, for which the estimation error dynamics is uniformly exponentially stable.

\section{Main Results -- Uniform Strong$^*$ Detectability}\label{sec:mainresult}
The concept presented in this section allows to asymptotically reconstruct the system states without explicitly differentiating the output signal in the presence of an unknown input $\w(t)$. 

Similar to the standard detectability notion, this concept is now introduced via the existence of a suitable observer and its corresponding error system.
\begin{definition}[uniform strong$^*$ detectability]\label{def:strongdetectability}
System~\eqref{eq:system} is called uniformly strong$^*$ detectable, if there exists an observer of the form
\begin{subequations}\label{eq:strongobsv}
	\begin{align}
	\dot \z &= \N(t)\z(t) + \R(t)\y(t)\label{eq:strongobsv_z} \\
	\hat \x &= \z + \S(t)\y \label{eq:strongobsv_x}
	\end{align}
\end{subequations}
with the estimate $\hat\x(t)\in\mathds R^n$, uniformly bounded matrices $\N(t), \R(t)$, and a uniformly bounded Lipschitz continuous matrix $\S(t)$ such that the estimation error $\e=\x-\hat\x$ satisfies
	\begin{equation}\label{eq:error:ues_bound}
		\|{\e(t)}\| \le M e^{-\mu (t - t_0)} \|{\e(t_0)}\|
	\end{equation}
for all $t\geq t_0$, $t_0 \in\mathds J$ and some constant $\mu$, $M>0$.
The observer~\eqref{eq:strongobsv} is then called a \emph{strong observer}.
\end{definition}

In order to eventually obtain conditions for uniform strong$^*$ detectability, the following lemma further specifies the properties of the observer~\eqref{eq:strongobsv}.
\begin{lemma}
	System~\eqref{eq:strongobsv} is a strong observer for~\eqref{eq:system}, if and only if the following relations are fulfilled simultaneously:
\begin{enumerate}
	\item[(r1)] $\dot\e = \N(t)\e$ is UES,
	\item[(r2)] $\N(t)-\A(t)+(\R(t)-\N(t)\S(t)+\dot\S(t))\C(t) + \S(t)(\C(t)\A(t)+\dot\C(t))=\vc 0$,
	\item[(r3)] $\D(t)-\S(t)\C(t)\D(t)=\vc 0$.
\end{enumerate}
\end{lemma}
\begin{proof}
In the following, the time dependency is omitted for a better readability.
Differentiating~\eqref{eq:strongobsv_x} gives
\begin{equation}
	\dot{\hat\x} = \dot \z + \dot\S \y +\S \dot \y
\end{equation}
and by using relation~\eqref{eq:strongobsv_z} one obtains
\begin{equation}
	\dot{\hat\x} = \N\hat\x + (\R+\dot\S -\N\S)\y + \S \dot \y.
\end{equation}
Due to Lipschitz continuity of $\S$, $\dot \S$ exists almost everywhere and is uniformly bounded. 
With $\y=\C\x$ and $\dot\y = (\C\A+\dot\C)\x+\C\D\w$, the error system $\dot \e = \dot \x - \dot{\hat \x}$ can be stated as
\begin{equation}
	\begin{aligned}
		\dot\e &= \dot \x -\dot{\hat\x} = \A\x+\D\w-\N\hat\x-(\R-\N\S+\dot\S)\y - \S\dot\y\\
		&= \A\x+\D\w-\N\hat\x -(\R-\N\S+\dot{\S})\C\x\\ &-\S(\C\A+\dot\C)\x-\S\C\D\w\\
	\end{aligned}
\end{equation}
Together with the relation $\hat\x =\x-\e$, this yields
\begin{equation}\label{eq:errorsys}
	\begin{aligned}
		\dot \e &= \N\e -\left[\N-\A+(\R-\N\S+\dot\S)\C + \S(\C\A+\dot\C)\right]\x \\ &+ (\D-\S\C\D)\w.
	\end{aligned}
\end{equation}	
Hence, necessity and sufficiency of the relations (r1)--(r3) follow directly.
\end{proof}
Up to now, it is not obvious how to design the observer~\eqref{eq:observer}.
The design can be formulated as a classical observer design problem based on the following main result.
\begin{theorem}\label{thm:equivalence}
	Under the assumptions (A1)--(A3), the following statements are equivalent
	\begin{enumerate}
		\item[\textit{(i)}] System~\eqref{eq:system} is uniformly strong$^*$ detectable.
		\item[\textit{(ii)}] The pair $\left(\tilde \A(t),\tilde \C(t)\right) $ with
		\begin{equation}\label{eq:Atilde}
		\begin{aligned}
		\tilde{\A}(t) &= \A(t)-\D(t)\mtGamma^\dag(t)\left[\C(t)\A(t)+\dot \C(t)\right], \quad \\ \tilde \C(t) &= \begin{bmatrix}
		\C(t)\\ \C(t)\tilde\A(t)+\dot \C(t)
		\end{bmatrix},
		\end{aligned}
		\end{equation}
		is uniformly exponentially detectable.
	\end{enumerate}
\end{theorem}
\begin{remark}
    The use of this equivalence for the purpose of observer design will be discussed in Section~\ref{sec:obsvdesign}.
\end{remark}
The equivalence is proven in the following two subsections.
\subsection{Item (i) implies item (ii)}
All matrices $\S(t)$ that fulfill (r3) can be parametrized as
\begin{equation}\label{eq:decoupling}
\S = \D \mtGamma^\dag + \L_2 \left[\I_p-\mtGamma\mtGamma^\dag\right]
\end{equation}
with any Lipschitz continuous $\L_2(t)\in\mathds{R}^{n\times p}$. 
Hence, for a given $\S(t)$, it is possible to obtain a (non-unique) $\L_2(t)$, such that~\eqref{eq:decoupling} is fulfilled.
Because $\mtGamma(t)$ has full rank and is Lipschitz continuous, the pseudo-inverse $\mtGamma^\dag(t)$ and its derivative exist and are uniformly bounded, see~\cite[Thm. 10.5.3]{campbell2009generalized}.

From (r2), it can then be concluded that
\begin{equation}
\N = \A- (\R-\N\S + \dot\S)\C - \S(\C\A + \dot \C)
\end{equation}
in combination with~\eqref{eq:decoupling} and 
\begin{equation}\label{eq:strongobsv:L1}
	\L_1 = \R-\N\S + \dot\S
\end{equation}
yields
\begin{equation}
\N = \A-\D\mtGamma^\dag(\C\A+\dot\C) - \L_1\C - \L_2(\I_p-\mtGamma\mtGamma^\dag)(\C\A+\dot\C).
\end{equation}
With $\mtGamma=\C\D$ and $\tilde \A$ as in~\eqref{eq:Atilde}, one obtains
\begin{align}
    \C \tilde\A + \dot \C &= \C \A + \dot \C - \C \D \mtGamma^\dag (\C \A + \dot \C) \nonumber \\
    &= \left[\I_p-\mtGamma\mtGamma^\dag\right]\left[\C\A+\dot\C\right].
\end{align}
From here, it follows that
\begin{equation}\label{eq:Nreduced}
\N = \tilde \A -\L_1\C -\L_2(\C\tilde\A+\dot \C).
\end{equation}
Because of (r1), it then can be concluded that 
\begin{equation}
\left(\tilde \A,\begin{bmatrix}
\C\\ \C\tilde\A+\dot \C
\end{bmatrix}\right)
\end{equation}
is uniformly exponentially detectable.
Boundedness of the gains $\L_1$ and $\L_2$ is guaranteed by the existence of observer~\eqref{eq:observer}.

\subsection{Item (ii) implies item (i)}\label{sec:proof:part2}
According to Lemma~\ref{le:detectability}, uniform detectability of $(\tilde \A(t),\tilde \C(t))$ is equivalent to the existence of a uniformly bounded solution $\P(t)\succeq 0$ to the Riccati equation
\begin{equation}\label{eq:strongobsv:Riccati}
	\dot \P = \tilde \A(t)\P + \P\tilde \A(t)^\transp -\P \tilde \C(t)^\transp \tilde \C(t) \P + \Q(t),\; \P(0)\succeq 0,
\end{equation}
where the positive definite matrix $\Q(t)\succ 0$ is considered as a tuning parameter.
Moreover, the system
\begin{equation}\label{eq:errorsystem:obsv}
	\dot\e = \left[ \tilde\A(t)-\P(t)\tilde \C^\transp(t)\tilde \C(t)\right]\e
\end{equation}
is uniformly exponentially stable.
The error system~\eqref{eq:errorsystem:obsv} is equivalent to $\dot\e=\N(t)\e$ with $\N(t)$ as in~\eqref{eq:Nreduced} and
\begin{equation}\label{eq:strongobsv:gains}
	\L_1(t)= \P(t) \C^\transp(t),\quad \L_2(t) = \P(t)(\C(t)\tilde\A(t)+\dot\C(t))^\transp.
\end{equation}
From~\eqref{eq:strongobsv:L1} it then follows that
\begin{equation}\label{eq:R}
\R(t) = \L_1(t) + \N(t)\S(t)- \dot\S(t),
\end{equation}
where $\S(t)$ is obtained from~\eqref{eq:decoupling} together with~\eqref{eq:strongobsv:gains} .
The derivative $\dot \S(t)$ is uniformly bounded because of (A1)-(A3) and the boundedness of the derivative of $\L_1(t)$ and $\L_2(t)$ in~\eqref{eq:strongobsv:gains}.
Therefore, item \textit{(ii)} is sufficient for the existence of an observer of the form~\eqref{eq:strongobsv}, and thus, following Definition~\ref{def:strongdetectability}, the system~\eqref{eq:system} is uniformly strong$^*$ detectable. \hfill \qed

\subsection{Observer design}
\label{sec:obsvdesign}
The results stated in Theorem~\ref{thm:equivalence} and in Section~\ref{sec:proof:part2} can be utilized for the observer design.
Due to Theorem~\ref{thm:equivalence}, it is possible to design an observer for an auxiliary system
\mbox{$	\dot {\bm \xi} = \tilde \A(t) {\bm \xi},\; \y_\xi =\tilde \C(t) {\bm \xi}$}
with $\tilde \A(t)$ and $\tilde \C(t)$ as in~\eqref{eq:Atilde}.
The observer gains $\L_1(t)$ and $\L_2(t)$ are obtained from the solution of the Riccati differential equation~\eqref{eq:strongobsv:Riccati} and~\eqref{eq:strongobsv:gains}.
For the implementation of the observer, one then uses the form~\eqref{eq:strongobsv} with
$\S(t)$, $\N(t)$ and $\R(t)$ given by~\eqref{eq:decoupling},~\eqref{eq:Nreduced} and \eqref{eq:R}, respectively.

\subsection{Discussion and Extensions}

This section discusses some consequences and extensions of the main results.

The following result shows that uniform strong$^*$ detectability is sufficient for strong$^*$ detectability as introduced in~\cite[Definition 1.3]{hautus1983strong} for linear time invariant systems.
\begin{corollary}
	If system~\eqref{eq:system} is uniformly strong$^*$ detectable, then, $\lim_{t\rightarrow\infty} \y(t)= \bf 0$ implies $\lim_{t\rightarrow\infty}\x(t)= \bf 0$, independent of the unknown input $\w(t)$.
\end{corollary}
The proof follows directly from the observer~\eqref{eq:strongobsv} and assuming $\y\rightarrow \bf 0$ for $t\rightarrow \infty$.
In the time invariant case, uniform strong$^*$ detectability is also necessary for strong$^*$ detectability, see~\cite{darouach1994fullorder} for the continuous time and ~\cite{valcher1999state} for the discrete time case. 
Similar to the standard detectability case, the converse is not true.
A counter-example, which satisfies assumptions (A1)--(A3), is given by the system
\begin{equation}
	\begin{aligned}
		\dot \x &= \begin{bmatrix}
			0 &0 \\ 0 &-\frac{1}{t}
		\end{bmatrix}\x + \begin{bmatrix}
			1\\ 0
		\end{bmatrix}w, \quad
		y = \begin{bmatrix}
			1 & 0 
		\end{bmatrix}\x,
	\end{aligned}
\end{equation}
for which $\lim_{t\rightarrow\infty}y(t) = 0$ implies $\lim_{t\rightarrow\infty}\x(t) = 0$.
The pair $(\tilde \A(t), \C(t))$ with $\tilde \A = \A$ is not detectable and hence no strong observer of the form~\eqref{eq:strongobsv} exists.

In some situations, the detectability condition $(ii)$ of Theorem~\ref{thm:equivalence} can be simplified, which is stated in the following corollary.
\begin{corollary}
	For $p=q$, the pair $\left(\tilde \A(t),\tilde \C(t)\right)$ is detectable if and only if the pair 
	$\left(\tilde \A(t),\C(t)\right)$ is detectable.
\end{corollary}
The corollary follows directly from the fact that the pseudo inverse reduces to the inverse for $p=n$ and hence $\C(t)\tilde\A(t)+\dot\C(t)=\vc 0$.
Detectability of $\left(\tilde \A(t),\C(t)\right)$ is a sufficient condition for the detectability of $\left(\tilde \A(t),\tilde\C(t) \right)$ also for $p>q$. 

In the time invariant case, the result can also be simplified:
\begin{corollary}\label{le:det:lti}
	For time invariant systems, the pair $(\tilde\A,\C)$  is detectable if and only if $(\tilde\A,\tilde\C)$ is detectable.
\end{corollary}
This result follows directly from the equivalence of the null spaces of the corresponding observability matrices.

\section{Numerical Example}\label{sec:numericalexample}
A linearized version of the Lorenz'96 model, see, e.g.~\cite{tranninger2020detectability,bocquet2017degenerate}, is investigated as a numerical example in this section. 
This nonlinear model proposed by E. Lorenz is widely used as a benchmark example in data-assimilation, see, e.g.,~\cite{bocquet2017degenerate,trevisan2011onthekalman}. 
The system is given by
	\begin{equation}\label{eq:lorenz96}
		\dot{z}_i=(z_{i+1}-z_{i-2})z_{i-1}-z_i + f_i(t),\; i=1,\ldots,n,
	\end{equation}
with $z_{-1}\coloneqq z_{n-1}$, $z_{0}\coloneqq z_n$, and \mbox{$z_{n+1}\coloneqq z_1$}. 
The state vector is $\z = [z_1\;\cdots \; z_n]^\transp\in\mathds{R}^n$. 
The system order and the output are chosen as $n=18$ and \begin{equation}\label{eq:lorez96output}
	\y(t) = \begin{bmatrix} \e_1 &
\e_5&
\e_9&
\e_{12}&	
\e_{15} \end{bmatrix}^\transp \z
	\end{equation}
respectively, where $\e_i$ is the $i$-th euclidian standard basis vector.
For a constant forcing $f_i(t)=f_0=8$, the system exhibits a chaotic behaviour~\cite{bocquet2017degenerate}.

The unknown input is assumed to act as a perturbation of the constant forcing $f_i(t)=f_0+d_i(t)w(t)$ with $w(t)$ as the scalar unknown input. 
The input matrix is chosen as 
$\D(t) = \left[0 \; \cdots \; 0 \quad d_5(t) \quad 0 \; \cdots\; 0 \quad d_{12}(t) \quad 0 \; \cdots \;\right]^\transp$
with $d_5(t)=\sin(t)$ and $d_{12}(t)=\cos(t)$. 
The unknown input is chosen as $w(t)=\exp(0.5t)$ and is thus unbounded.
The approaches presented in~\cite{galvan-guerra2016highorder,tranninger2018sliding} are hence not applicable.

The linear time varying system in the form of~\eqref{eq:system} is obtained by linearizing ~\eqref{eq:lorenz96} and~\eqref{eq:lorez96output} along a trajectory.
Since no analytical solution for the nonlinear system is available, a numerical solution was computed for the unperturbed case, i.e., $\w= 0$, with the initial conditions $z_{i}(0)= \sin\left(\frac{i-1}{n}2\pi\right)$ for $i=1,\ldots,n$. 

In order to conduct a detectability analysis of the pair $\left(\tilde\A(t),\C(t) \right)$, simulation studies were performed.
These showed that the pair is detectable, i.e., the solution of the corresponding Riccati equation~\eqref{eq:strongobsv:Riccati} remains uniformly bounded.
This simplifies the observer design, because $\tilde \C(t)$ can be replaced with $\C(t)$ in equation~\eqref{eq:strongobsv:Riccati} and hence $\L_2(t)=\bf 0$.
For $\P(0)=\I_n$ and $\Q(t)=10\I_n$, the minimum and maximum singular value of $\P(t)$ are shown in~Fig.\ref{fig:sigmap}. 

The initial estimation error is chosen randomly from a uniform distribution in the interval $(-1,1)$.
The norm of the estimation error is depicted in Fig.~\ref{fig:logerrornorm} and components of the estimation error are shown in Fig.~\ref{fig:errors}.
One can see that the error norm decays exponentially and uniform exponential stability of the error system is verified with the simulation results. 
Moreover, the error system is independent of the (unbounded) unknown input, as predicted by the theoretical findings.

\begin{figure}[tbp]
	\centering
	\includegraphics[width=0.95\linewidth]{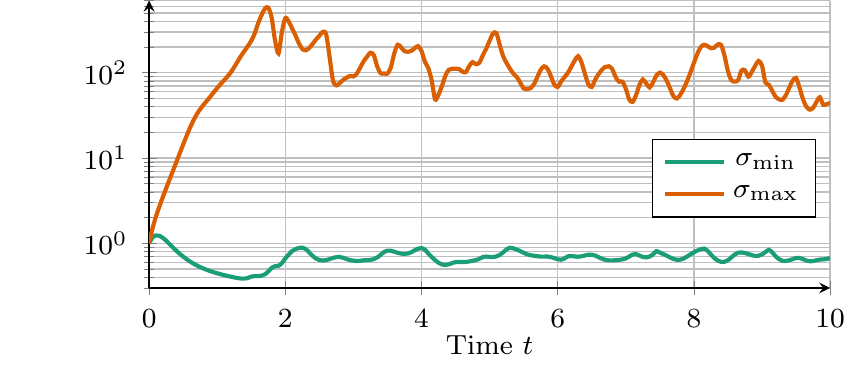}
	\caption{Minimum and maximum singular value of $\P(t)$.}
	\label{fig:sigmap}
\end{figure}

\begin{figure}[tbp]
	\centering
	\includegraphics[width=0.95\linewidth]{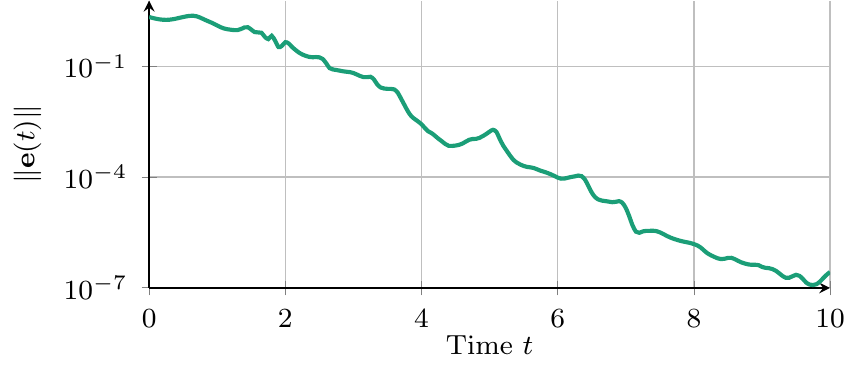}
	\caption{Two-norm of the estimation error.}
	\label{fig:logerrornorm}
\end{figure}

\begin{figure}[tbp]
	\centering
	\includegraphics[width=0.95\linewidth]{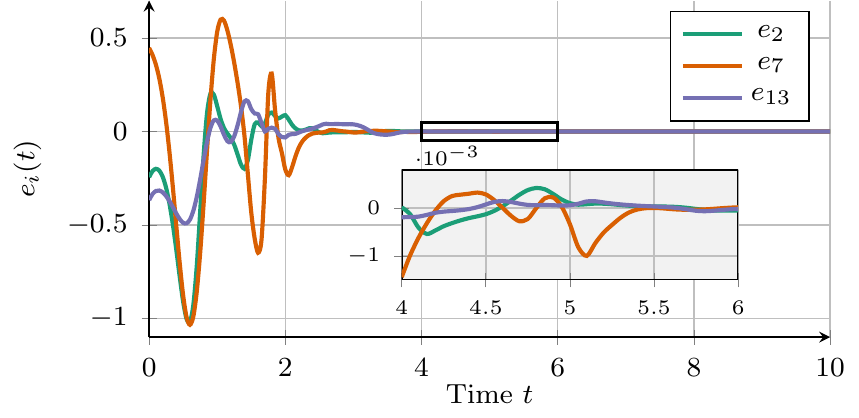}
	\caption{Components of the estimation error.}
	\label{fig:errors}
\end{figure}

\section{Discussion and Outlook}\label{sec:discussion}
The present work introduces a uniform strong detectability notion and an observer design procedure for linear time varying systems.
The proposed observer design allows to asymptotically reconstruct the system states in the presence of arbitrary unknown inputs.
In future, this could also be studied in combination with reduced observer design methods as presented in~\cite{tranninger2020detectability}.
It would be interesting to link uniform strong$^*$ detectability with the time varying zero dynamics notion introduced in~\cite{ilchmann2007time-varying} or with time varying characterizations of system zeros in order to obtain additional insights as in the time invariant case.

\appendix

\balance%
\bibliographystyle{elsarticle-num}
\bibliography{bibliography} 

\end{document}